\documentclass[aps,superscriptaddress,twocolumn,twoside,floatfix,pra,nofootinbib,a4paper]{revtex4-2}

\usepackage{amssymb}

\usepackage{dcolumn}
\usepackage{bm}
\usepackage{dsfont}
\usepackage{graphicx,epsfig}
\usepackage{amsmath}
\usepackage{braket}
\usepackage{physics}
\usepackage{caption}
\usepackage{float}
\usepackage[shortlabels]{enumitem}
\usepackage{amsthm}
\usepackage{soul}

\usepackage{blindtext}
\usepackage{color}

\usepackage[colorlinks]{hyperref}
\hypersetup{
	colorlinks = true,
	urlcolor = {blue},
	citecolor = {magenta},
	linkcolor= {blue}
}

\renewcommand{\eqref}[1]{Eq.~(\ref{#1})}
\newcommand{\figref}[1]{Fig.~\ref{#1}}

\newcommand{\appref}[1]{App.~\ref{#1}}
\newcommand{\secref}[1]{Sec.~\ref{#1}}

\newtheorem{observation}{Observation}
\newtheorem{result}{Result}

\newtheoremstyle{normaldef}  
  {3pt}   
  {3pt}   
  {}      
  {}      
  {\normalfont} 
  {.}     
  { }     
  {}      

\theoremstyle{normaldef}

\begin{document}

\title{
Characterising the set of deterministic quantum correlations\\ 
in prepare-and-measure scenarios
}

\author{Nicola D'Alessandro}
\address{Physics Department and NanoLund, Lund University, Box 118, 22100 Lund, Sweden.}

\author{Oliver Karlsson}
\address{Physics Department and NanoLund, Lund University, Box 118, 22100 Lund, Sweden.}

\author{Carles Roch i Carceller}\email{carles.roch@icfo.eu}
\address{Physics Department and NanoLund, Lund University, Box 118, 22100 Lund, Sweden.}
\address{ICFO - Institut de Ciencies Fotoniques, The Barcelona Institute of Science and Technology, 08860 Castelldefels, Spain.}

\begin{abstract}

Correlations that do not admit a deterministic explanation are a central feature of quantum theory and a key resource for quantum information processing. Identifying and certifying such correlations, however, remains a fundamental challenge. In this work, we advance on this problem by considering deterministic correlations in the prepare-and-measure scenario consistent with a wide range of communication restrictions. To certify correlations incompatible with such deterministic models, we ask whether they admit a scenario in which measurement outcomes can be perfectly predicted by an adversary equipped with classical side-information. We show the usefulness of this approach and propose semidefinite programming relaxations tailored to three representative communication restrictions: fixed ensemble, bounded overlaps and restricted observables.

\end{abstract}

\maketitle

\section{Introduction}

Quantum theory predicts correlations that admit no deterministic explanation: measurement outcomes cannot, in general, be understood as revealing pre-existing properties of physical systems. This intrinsic non-determinism is a hallmark of non-classicality, and a key resource for quantum information processing. Characterizing the boundary between deterministic and quantum correlations is therefore a central problem in quantum information science.

This question has been extensively investigated in Bell scenarios, where spatially separated parties generate correlations that cannot be explained by local hidden-variable models \cite{Bell1964,CHSH1969}. However, Bell experiments require demanding resources, including entanglement distribution and loophole-free implementations \cite{Larsson2014}. An alternative framework is provided by the prepare-and-measure scenario, in which a system is prepared and subsequently measured \cite{brask2026}. Despite their simplicity, such scenarios can reveal correlations that cannot be reproduced by classical models under suitable communication constraints \cite{spekkens2005,gallego2010}. This has made them a versatile platform for a wide list of applications, including quantum communication \cite{pauwels2025}, randomness certification \cite{li2011} and self-testing \cite{tavakoli2018,navascues2023}.

The role of non-determinism in the prepare-and-measure scenario is closely connected to its role in Bell non-locality. Any deterministic non-signaling behavior necessarily admits a local hidden-variable model \cite{Jarrett1984,Shimony1993,popescu1994}. Consequently, every non-local behavior is intrinsically non-deterministic, which motivates the certification of non-classicality as the exclusion of deterministic explanations.

Drawing inspiration from the above, we introduce an approach for certifying non-classicality in the prepare-and-measure scenario by asking whether the observed correlations admit a deterministic explanation under certain communication restrictions. To address this question, we introduce an adversary equipped with classical side-information. Then, we say that the observable correlations admit a deterministic model compatible with the given communication restriction whenever the adversary is able to perfectly predict the value of the measurement outcomes.


The certification approach we employ is complementary to quantum random number generation (QRNG). Rather than quantifying the unpredictability compatible with observed correlations, we ask whether those correlations admit a model with predictable outcomes. Existing approaches for semi-device-independent QRNG, which can be treated with efficient numerical techniques based on semidefinite programming (SDP) \cite{tavakoli2024}, can therefore be adapted into feasibility tests for deterministic explanations. This provides a systematic route to study determinism beyond standard dimension restrictions. We demonstrate its usefulness with three representative communication restrictions: fixed ensemble, bounded overlaps and restricted observables. Finally, we employ the adversarial perspective to derive analytical bounds for generic linear witnesses.

\section{The prepare-and-measure scenario}

Consider an experiment involving two parties: Alice and Bob. In each round, a shared random variable $\lambda$ is distributed with probability $q(\lambda)$. Alice then selects an input $x\in X$, prepares a quantum state $\rho_x^\lambda$, and sends it to Bob. Bob then chooses a measurement setting $y\in Y$ and performs a measurement represented by the positive operator-valued measure (POVM) $\{M_{b|y}^\lambda\}_{b}$, producing an outcome $b\in \mathcal{B}$. The observable events in the experiment averaged over $\lambda$ are specified by the triplets $(b,x,y)$, which are collected over many rounds to define the conditional probabilities,
\begin{align}
    p(b|x,y) = \sum_\lambda q(\lambda)\tr\!\left(\rho_x^\lambda M_{b|y}^\lambda\right) .
\end{align}
It is well known that certain conditional distributions $p(b|x,y)$ can exhibit genuinely quantum features with no classical analogue. The standard notion of classicality is based on models in which Alice is restricted to sending a classical message $m=1,2,\ldots,d$. The resulting correlations take the form
\begin{align}\label{eq:SC}
    p(b|x,y) = \sum_\lambda q(\lambda) \sum_{m=1}^{d} p_\lambda(m|x) p_\lambda(b|m,y) .
\end{align}
The relevance of the classical decomposition in \eqref{eq:SC} stems from constraining the message alphabet size $d$ \cite{gallego2010}. Classicality has also been investigated under alternative communication constraints, including energy \cite{VanHimbeeck2017}, information content \cite{tavakoli2022quantum} and non-contextuality \cite{schmid2018,flatt2022}. However, these approaches are tailored to specific scenarios and lack a systematic general framework.

Here, we aim to certify whether the observable correlations admit a deterministic model consistent with any given communication restriction. The approach we use is based on a task complementary to semi-device-independent QRNGs: given $p(b|x,y)$ and a communication restriction, does there exist a model in which the outcomes can be perfectly predicted? If such a model exists, the correlations admit a deterministic explanation.

A key advantage of this perspective is that it accommodates a wide range of communication restrictions, untied to specific settings. Recent advances in semi-device-independent QRNG already considered a broad range of assumptions, including restricted dimension \cite{li2011}, state overlaps \cite{brask2017}, target-state fidelities \cite{tavakoli2021}, energy constraints \cite{vanhimbeeck2019}, photon-number components \cite{carceller2025_photon}, rotational symmetries \cite{aloy2024}, information-content \cite{tavakoli2020}, non-contextuality \cite{carceller2022} and quantum speed limits \cite{jones2026}. 

Motivated by these developments, here we develop an approach for certifying non-classicality that works under any relevant communication restriction. For illustration purposes however, we consider only three representative examples. First, we assume that Alice prepares the known ensemble $\mathcal{E}=\{\rho_x\}$ on average, such that $\rho_x=\sum_\lambda q(\lambda)\rho_x^\lambda$ $\forall x$. Second, we consider a weaker characterization in which the preparations are constrained only through lower-bounds on their pairwise overlaps, with $\rho_x^\lambda=\ketbra*{\psi_x^\lambda}{\psi_x^\lambda}$ such that $\sum_\lambda q(\lambda) \abs*{\braket*{\psi_{x}^\lambda}{\psi_{x'}^\lambda}} \geq \beta_{xx'}$. Finally, we consider more versatile frameworks where an observable determined by the set of hermitian operators $\{O_{ix}\}_{x,i}$ is bounded by $\sum_\lambda q(\lambda) \tr(\rho_x^\lambda O_{ix})\leq \omega_{ix}$.

The three assumptions we consider are summarized as follows:

\begin{list}{}{
  \setlength{\leftmargin}{1.1cm}
  \setlength{\labelwidth}{1cm}
  \setlength{\labelsep}{0.2cm}
  \setlength{\itemindent}{0cm}
  \setlength{\itemsep}{0.2cm}
}
    \item[$\mathcal{S}_{\mathcal{E}}$:] Fixed ensemble:
    $\mathcal{E}=\{\rho_x\}$, with $\rho_x=\sum_\lambda q(\lambda)\rho_x^\lambda$.
    
    \item[$\mathcal{S}_{\beta}$:] Bounded overlaps: $\sum_\lambda q(\lambda) \ \abs*{\!\braket*{\psi_x^\lambda}{\psi_{x'}^\lambda}}\geq \beta_{xx'}$.

    \item[$\mathcal{S}_{O}$:] Restricted observable: $\sum_\lambda q(\lambda) \tr(\rho_x^\lambda O_{ix})\leq \omega_{ix}$.
    
\end{list}
Let us collect all assumptions in the list $\mathcal{S}=\{\mathcal{S}_{\mathcal{E}},\mathcal{S}_{\beta},\mathcal{S}_{O}\}$. We define the set of quantum correlations 
\begin{align}\label{eq:Q_S}
\mathcal{Q}_{\mathcal{S}_i} = \left\{ p(b|x,y)=\sum_\lambda q(\lambda) \tr\!\left(\rho_x^\lambda M_{b|y}^{\lambda}\right) \mid  \mathcal{S}_i \right\} \ ,
\end{align}
which are compatible with the restriction $\mathcal{S}_i\in\mathcal{S}$.
The notion of classicality employed in this work is based on the emergence of determinism for a particular sub-set of inputs $X^\ast \subseteq X$ and $Y^\ast \subseteq Y$. More concretely, we say that the observable correlation $p\in\mathcal Q_{\mathcal{S}_i}$ is classical for $x\in X^\ast$ and $y\in Y^\ast$ if the single-round conditional probabilities admit a deterministic explanation, i.e.~$\tr\small(\rho_x^\lambda M_{b|y}^{\lambda}\small)=D_{\lambda}(b|x,y)$, where $D_{\lambda}(b|x,y)=\{0,1\}$. With that, we define the set of classical correlations associated to the inputs $(x,y)$ as,
\begin{align}
\mathcal{C}_{\mathcal{S}_i}^{(x,y)} = \Biggl\{ \sum_\lambda q(\lambda) D_\lambda(b|x,y) \mid D_\lambda(b|x,y) \in \mathcal{Q}_{\mathcal{S}_i} \Biggr\} . 
\end{align}
It is not difficult to see that any behavior satisfying $p(b|x,y)\in \mathcal{C}_{\mathcal{S}_i}^{(x,y)}$, $\forall x,y\in X,Y$, will always admit the decomposition from \eqref{eq:SC}. The converse, however, is only true when the Hilbert space dimension becomes the relevant constraint. The set $\mathcal{C}_{\mathcal{S}_i}^{(x,y)}$ is, however, not bounded to dimensional assumptions, and may be considered as a generalized notion of classicality compatible with alternative communication restrictions. 

\begin{figure}
\includegraphics[width=0.45\textwidth]{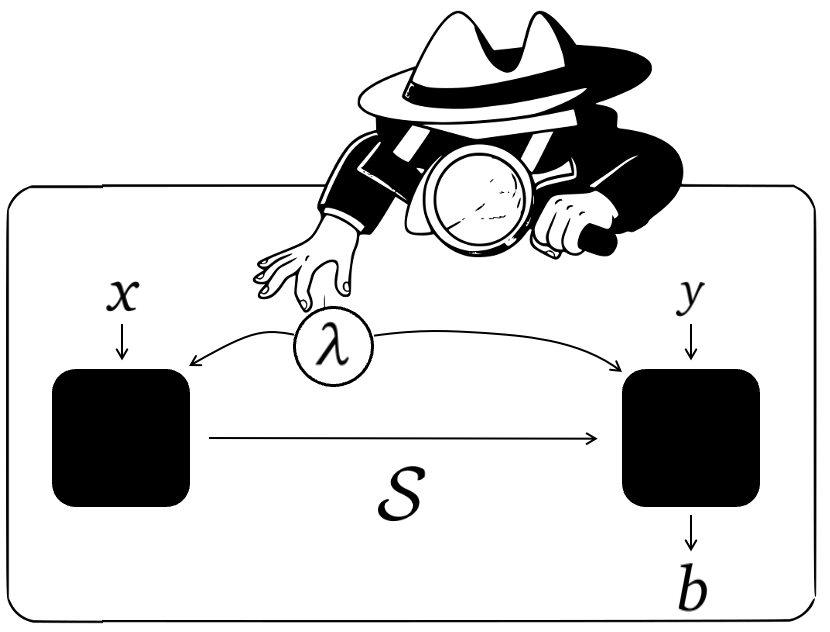}
\caption{\textit{Adversarial prepare-and-measure scenario.} The adversary is correlated through the hidden-variable $\lambda$ distributed among the devices. The communication is bound to satisfy the restriction $\mathcal{S}$.}
\label{fig:pam_spy}
\end{figure}

\section{Adversarial scenario}

We now introduce an adversary, Eve, with the goal of guessing the measurement outcomes, see \figref{fig:pam_spy}. Eve is assumed to have complete knowledge of the random variable $\lambda$, formalized as follows. We may assume, without loss of generality, that the preparations in each round are pure, namely $\rho_x^\lambda=\ketbra*{\psi_x^\lambda}{\psi_x^\lambda}$, as any mixture can be absorbed by $\lambda$. With that, we expand the Hilbert space to define the purified quantum state $\ket{\Psi_x} = \sum_\lambda \sqrt{q(\lambda)} \ket{\psi_x^\lambda} \otimes \ket{\lambda} \otimes \ket{\lambda}_{\mathcal{H}_E}$ and measurements $B_{b|y} = \sum_\lambda M_{b|y}^{\lambda} \otimes \ketbra{\lambda}{\lambda}$, where $\mathcal{H}_E$ is the dilated Hilbert space accessible to Eve. This way, Eve may perform a measurement $\{E_\lambda=\ketbra{\lambda}{\lambda}\}_\lambda$ in $\mathcal{H}_E$ to obtain full knowledge of the random variable $\lambda$. The probability that Eve and Bob respectively observe $\lambda$ and $b$, given the choice of preparation and measurement settings $(x,y)$ is given by,
\begin{align}\label{eq:glob_corr_pam}
p(\lambda,b|x,y) \!= \bra{\Psi_x}\!B_{b|y}\!\otimes\! E_{\lambda}\!\ket{\Psi_x} \!=\! q(\lambda) \tr\!\left(\rho_x^\lambda M_{b|y}^{\lambda}\right) . 
\end{align}
The goal of Eve is to predict Bob's measurement outcome using her knowledge of the shared randomness. The probability that Eve succeeds in this task is obtained by optimizing over all realizations compatible with $\mathcal{S}$ that reproduce the observed distribution $p(b|x,y)$. For a fixed set of inputs $(x,y)$, this is given by the optimal guessing probability 
\begin{align}\label{eq:pg_ensemble}
p_g^{(x,y)} = \max_{p\in \mathcal{Q}_{\mathcal{S}_i}} \sum_\lambda q(\lambda) \max_b\!\left\{p_\lambda(b|x,y)\right\},
\end{align}
where $p_\lambda(b|x,y)$ denotes the observable conditional probabilities for a concrete choice of $\lambda$. If the observable behavior admits a classical explanation for a particular choice of inputs $(x,y)$, then Bob's measurement outcomes must be perfectly predictable by any adversary with access to the shared randomness that governs the underlying deterministic mechanism. Our goal, however, is to use a complementary argument to detect non-classical correlations, summarized in the following observation.

\begin{observation}[Predictable is classical] \label{obs:pg_det}
The correlations produced in the prepare-and-measure scenario for a concrete choice of input settings $(x,y)$ admit a classical model iff $p_g^{(x,y)}=1$.
\end{observation}

\begin{proof}
That $p(b|x,y)\in\mathcal{C}_{\mathcal{S}_i}^{(x,y)}$ guarantees $p_g^{(x,y)}=1$. What we need to show thus is the converse. But indeed, the guessing probability can only be one if $\max_b\left\{p_\lambda(b|x,y)\right\}=1$, $\forall \lambda$. This implies that the distribution $p_\lambda(b|x,y)$ corresponds to a deterministic assignment, thereby rendering $p(b|x,y)\in\mathcal{C}_{\mathcal{S}_i}^{(x,y)}$.
\end{proof}

\section{Certification methods}

By Observation~\ref{obs:pg_det}, certifying non-classicality becomes equivalent to excluding models that reproduce the observed behavior with predictable measurement outcomes. This complementarity with randomness certification enables existing QRNG techniques to be adapted for non-classicality detection under various communication constraints. Here, we explore this idea with the three restrictions $\mathcal{S}=\{\mathcal{S}_{\mathcal{E}},\mathcal{S}_{\beta},\mathcal{S}_{O}\}$ introduced above. We outline three search problems in the form of SDPs tailored to each restriction, whose infeasibility certifies non-classicality for the sub-sets inputs $X^\ast\subseteq X$ and $Y^\ast\subseteq Y$.

\subsection{Fixed ensemble $(\mathcal{S}_{\mathcal E})$}

To determine whether the observed correlations $p(b|x,y)$ generated by the ensemble $\mathcal E=\{\rho_x\}$ admit a classical explanation, one needs to search for a decomposition of the ensemble into hidden-variable states $\{\rho_x^\lambda\}$ that reproduce the observed correlations with predictable outcomes. To address this problem we revisit the block moment matrix (BMM) approach introduced in Ref.~\cite{dalessandro2026}. Let $L^\lambda=\{\mathds{1},\{\rho_x^\lambda\}_x,\{M_{b|y}^\lambda\}_{b,y}\}$ be the set of relevant operators for each $\lambda$, and let $S^\lambda\supseteq L^\lambda$ contain products of elements of $L^\lambda$. With that, we construct the BMM with entries $\Gamma^\lambda_{u,v}=q(\lambda)\,u v^\dagger$, for $u,v\in S^\lambda$. Imposing the operator identities associated with the underlying quantum model yields an SDP relaxation of the decision problem. As we show in \appref{app:ensemble}, this problem reduces to the feasibility program
\begin{align}
\text{find} & \quad \{\Gamma^{\bm{\lambda}}\}_{\bm{\lambda}} \label{eq:SDPrelax_ensemble} \\
\text{s.t.} & \quad \Gamma^{\bm{\lambda}}\succeq 0, \quad \sum_{\bm{\lambda}}\Gamma^{\bm{\lambda}}_{\rho_x,\mathds{1}}=\rho_x \nonumber \\
& \quad p(b|x,y)=\sum_{\bm{\lambda}} \tr(\Gamma^{\bm{\lambda}}_{\rho_x, M_{b|y}})\nonumber \\
& \quad \sum_{\bm{\lambda}} \tr(\Gamma^{\bm{\lambda}}_{\rho_x, M_{\lambda_{xy}|y}})=1, \quad x\in X^\ast, \ y\in Y^\ast \nonumber ,
\end{align}
where $\bm{\lambda}=\{\lambda_{xy}\}_{x,y}$ indexes all $|\mathcal B|^{|X^\ast||Y^\ast|}$ deterministic assignments specifying the outcome $b=\lambda_{xy}$ for each $x\in X^\ast$ and $y\in Y^\ast$.

\subsection{Bounded overlaps $(\mathcal{S}_{\beta})$}

We now consider the case in which the preparations are characterized only through lower-bounds on their pairwise overlaps. To address this, we employ the Gram matrix SDP hierarchy introduced in Ref.~\cite{wang2019}. Let $S=\{\ket{\Psi_x},\,B_{b|y}\ket{\Psi_x},\,E_\lambda\ket{\Psi_x}\}$, with $[B_{b|y},E_\lambda]=0$, and define the Gram matrix $G$ through the entries $G_{u,v}=\braket{u}{v}$, for $\ket{u},\ket{v}\in S$. As shown in \appref{app:overlaps}, the overlap assumptions, the observed correlations, and the predictability constraints all become linear relations on $G$. The decision problem can thus be rendered through the SDP relaxation
\begin{align}\label{eq:SDP_gamm_det}
\text{find} & \quad G \\
\text{s.t.} & \quad G\succeq 0, \quad G_{\ket{\Psi_x},\ket{\Psi_{x'}}} \geq \beta_{xx'} \nonumber\\
& \quad p(b|x,y) = G_{\ket{\Psi_x},B_{b|y}\ket{\Psi_x}} \nonumber\\
& \quad \sum_b G_{E_b\ket{\Psi_x},B_{b|y}\ket{\Psi_x}}=1, \quad x\in X^\ast, y\in Y^\ast \nonumber ,
\end{align}
where the precise construction of the Gram matrix is detailed in \appref{app:overlaps}.

\subsection{Restricted observables $(\mathcal{S}_{O})$}

Finally, we consider scenarios characterized by constraints on observable expectation values. Such assumptions can be incorporated through the tracial moment-matrix SDP hierarchy employed in Refs.~\cite{pauwels2022,carceller2025_photon}. Let $L^\lambda=\{\mathds{1},\{\rho_x^\lambda\}_x,\{M_{b|y}^\lambda\}_{b,y},\{O_{ix}\}_{i,x}\}$ and let $S^\lambda\supseteq L^\lambda$ contain products of elements of $L^\lambda$. We introduce the moment matrices with entries $\Upsilon^\lambda_{u,v} = q(\lambda) \tr(uv^\dagger)$, for $u,v\in S^\lambda$. The operator relations defining the quantum model translate into linear constraints on the matrix entries. In particular, the observed statistics and observable bounds are represented by linear relations on elements in $\Upsilon^\lambda$. 
As we show in \appref{app:observables}, the resulting feasibility problem becomes
\begin{align}
\text{find} & \quad \{\Upsilon^{\bm{\lambda}}\}_{\bm{\lambda}} \label{eq:SDPrelax_obs} \\
\text{s.t.} & \quad \Upsilon^{\bm{\lambda}}\succeq 0, \quad  \sum_{\bm{\lambda}}\Upsilon^{\bm{\lambda}}_{\rho_x,O_{ix}}\leq\omega_{ix}, \nonumber \\
& \quad p(b|x,y)=\sum_{\bm{\lambda}} \Upsilon^{\bm{\lambda}}_{\rho_x, M_{b|y}}, \nonumber \\
& \quad \sum_{\bm{\lambda}} \Upsilon^{\bm{\lambda}}_{\rho_x, M_{\lambda_{xy}|y}} = 1, \quad x\in X^\ast, \ y\in Y^\ast \nonumber ,
\end{align}
where $\{\Upsilon^{\bm{\lambda}}\}_{\bm{\lambda}}$ is a family of moment matrices consistent with the underlying operator relations.

\begin{figure*}
    \centering
    \includegraphics[width=\linewidth]{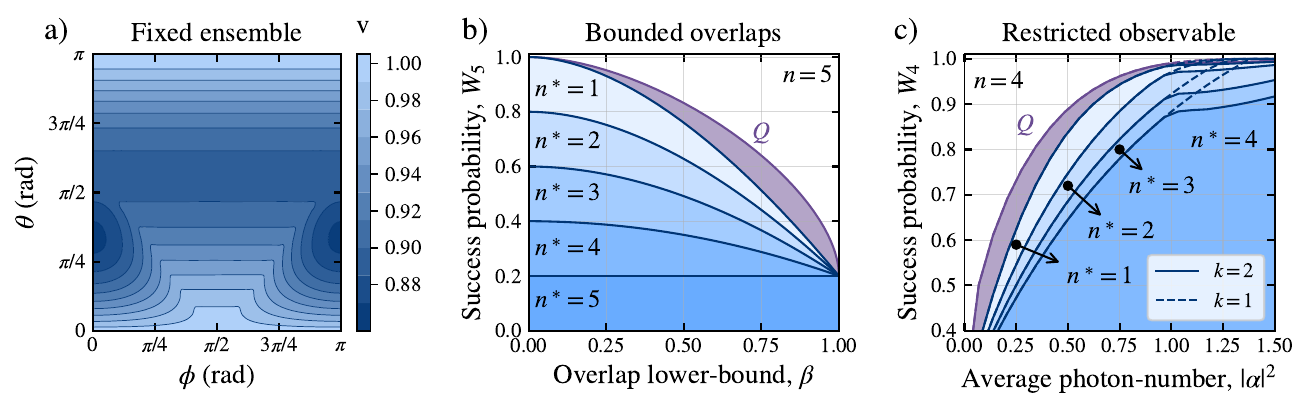}
    \caption{\textit{Classical models based on predictable measurement outcomes.} \textbf{a)} Critical visibility above which two qubit states cannot generate classical correlations. \textbf{b)} Value of $W_5$ above which a classical model is ruled out for $n^\ast$ inputs and preparations with bounded pair-wise overlaps. \textbf{c)} Value of $W_4$ above which a classical model is ruled out for $n^\ast$ inputs, and preparations with $k=1$ and $k=2$ photon-number components restricted by Poissonian statistics.}
    \label{fig:ENS_QSD_PN}
\end{figure*}

\section{Applications and results}

Equipped with the certification methods tailored to each communication restriction, we now assess their usefulness in various representative examples.

\subsection{Classicality with a fixed ensemble} 

The strength of the correlations in the prepare-and-measure scenario is often quantified through communication tasks whose performance is described by a linear witness,
\begin{align}\label{eq:pam_witness}
    W=\sum_{b,x,y}c_{bxy}p(b|x,y).
\end{align}
The coefficients $c_{bxy}$ specify the task under consideration. We begin by deriving a general bound on witnesses of the form \eqref{eq:pam_witness} under the assumption that the observed correlations admit a deterministic model for every preparation and measurement choice.

\begin{result}\label{res:max_W_E}
    Consider the linear witness in \eqref{eq:pam_witness}.
    For any ensemble of pure states $\mathcal{E}=\{\rho_x\}_x$, let the non-orthogonality graph have the preparations as vertices, with different preparations $x$ and $x'$ connected whenever $\rho_x$ and $\rho_{x'}$ are non-orthogonal. If the graph is connected, then any correlations $p(b|x,y)$ that admit a deterministic model satisfy
    \begin{align}
        W \leq \sum_{y} \max_{b}\left\{\sum_{x} c_{bxy}\right\} .
    \end{align}
\end{result}


\begin{proof}
    From Observation \ref{obs:pg_det}, classical behaviors for $x,y$ correspond to $p_g^{(x,y)} = 1$. For that to hold $\forall x,y$, we introduce the deterministic assignments $\bm{\lambda}:=\{\lambda_{xy}\}_{xy}$ which specify the outcomes $b=\lambda_{xy}$ for any input choice. Since the ensemble is composed of pure states, the preparation associated with each $\bm{\lambda}$ coincides with $\rho_x$. The condition $p_g^{(x,y)} = 1$ $\forall x,y$ thus implies that for every $\bm{\lambda}$ in the support of $q(\bm{\lambda})$, $\tr\small(\rho_x M_{\lambda_{xy}|y}^{\bm{\lambda}}\small) = 1$. 
    Such a decomposition is only possible if the hidden-variable assignments does not distinguish between preparation corresponding to non-orthogonal states. Thus, for every such pair $x,x'$, $\lambda_{xy} = \lambda_{x'y}$, $\forall y$. Since the non-orthogonality graph is connected, the equality propagates in the graph and the argument extends to all $x,x'$ pairs, including orthogonal ones. Hence $\lambda_{xy} = \lambda_{x'y} := \lambda_y$, $\forall x,x',y$. Substituting this restriction into the witness $W$, one obtains $\tr \small(\rho_x M_{b|y}^{\bm{\lambda}} \small) = \delta_{b, \lambda_{y}}$. The witness simplifies to $W = \sum_{\bm{\lambda}} q(\bm{\lambda}) \sum_{x,y} c_{\lambda_{y},x,y} \leq \max_{\bm{\lambda}} \sum_{x,y} c_{\lambda_{y},x,y}=\sum_y \max_{\lambda_y} \sum_{x} c_{\lambda_{y},x,y}$. Re-labeling $\lambda_y=b$, yields the claimed bound.
\end{proof}

We illustrate Result \ref{res:max_W_E} using a witness constructed from mutually unbiased bases (MUBs). Consider a preparation device producing states $\{\rho_{x_0x_1}\}$, where $x_1\in\{1,\ldots,n\}$ labels the basis and $x_0\in\{1,\ldots,d\}$ labels the basis element. The preparations satisfy $\tr(\rho_{x_0 x_1}\rho_{x'_0 x'_1}) = \delta_{x_0,x'_0}\frac{1}{d}$ for $x_1 \neq x'_1$. Similarly, let $\{M_{b|y}\}$ denote a collection of MUB measurements, where $y$ labels the basis and $b$ the measurement outcome. We consider the witness
\begin{align}
    W_{n,d} = \sum_{x_0=1}^{d} \sum_{x_1=1}^{n} \operatorname{tr}\!\left(\rho_{x_0 x_1} M_{x_0|x_1}\right),
\end{align}
corresponding to the coefficients $c_{bxy}=\delta_{b,x_0}\delta_{y,x_1}$. The quantum maximum is attained when preparations and measurements correspond to the same set of MUBs, yielding $W^{Q}_{n,d} = nd$. In contrast, applying Result \ref{res:max_W_E} immediately gives $W^{\rm Cl}_{n,d} = n$, which is strictly smaller than the quantum value for every $d>1$. Hence, the witness certifies non-classicality whenever $W_{n,d}>n$, which can be reached with detectors with an efficiency above $\eta^\ast = \frac{1}{d}$.

We complement our analysis by numerically determining the white-noise critical visibilities, $v\rho_x+\frac{1-v}{d}\mathds{1}$, above which the qubit ensembles $\rho_x=\ketbra*{\psi_x}{\psi_x}$ with $\ket{\psi_1}=\ket{0}$ and $\ket{\psi_2}=\sin\frac{\theta}{2}\ket{0}+e^{i\phi}\cos\frac{\theta}{2}\ket{1}$, measured along the $\sigma_x$ and $\sigma_z$ directions, does not admit classical models. This is done by solving the SDP relaxation in \eqref{eq:SDPrelax_ensemble} for all choices of input subsets $X^\ast$ and $Y^\ast$, retaining the lowest value. The resulting critical visibilities are shown in \figref{fig:ENS_QSD_PN}. \\

\subsection{Classicality with bounded overlaps}

We now turn to scenarios in which  the overlaps of the prepared states are symmetrically bounded by $\abs{\braket{\psi_x}{\psi_{x'}}} \geq \beta$,  $\forall x,x'$. We consider a measurement device with a single setting ($y=0$) and construct the state-discrimination witness $W_n=\frac{1}{n}\sum_x p(x|x,0)$. 
The corresponding deterministic bounds are obtained by finding the maximum $W_n$ for which the SDP in \eqref{eq:SDP_gamm_det} renders a feasible solution, under the assumption that at most $\abs{X^\ast} = n^\ast$ inputs contribute deterministically.

As an illustration, we focus on $n=5$. \figref{fig:ENS_QSD_PN} shows the optimal discrimination probability as a function of the overlap parameter for different values of $n^\ast$. Owing to the symmetry of the constraints, the bounds depend only on the cardinality of the subset $X^\ast$, and not on its specific choice.

In \appref{app.qsd}, we further analyze this optimization problem and derive a closed-form expression for the case $n^\ast=1$, 
\begin{align}\label{eq:det_bound_sd}
W_{n} \leq \frac{1}{n}\left(1+\frac{\Delta^2}{n-1}\right) 
\end{align}
where $\Delta = \sqrt{1-\beta}\left((n-2)+\sqrt{1+(n-1)\beta}\right)$. 
Curiously, the bound in \eqref{eq:det_bound_sd} for $n=2$ coincides with the maximum success probability compatible with a non-contextual model \cite{carceller2024}. \\

\subsection{Classicality with respect to an observable} 

Finally, we demonstrate the versatility of our approach under observable constraints applying it to the framework with photon-number restrictions developed in Ref.~\cite{carceller2025_photon}. The relevant observable in this case corresponds to $O_{ix}=\mathds{1}-\ketbra*{i}{i}$ in the photon-number basis. We consider $n$ coherent-state preparations with equal average photon-number $\abs{\alpha}^2$, where $\bra{i}\rho_x\ket{i} \geq e^{-\abs{\alpha}^2}\frac{\abs{\alpha}^{2i}}{i!}$ is assumed to follow Poissonian statistics. 

Under these assumptions, we again consider the state-discrimination witness $W_n$ and compute its maximal value compatible with the feasibility of the SDP in \eqref{eq:SDPrelax_obs}, allowing for at most $|X^\ast|=n^\ast$ state-inputs that contribute deterministically. We illustrate the obtained bounds in \figref{fig:ENS_QSD_PN} by choosing $n=4$ and restricting up to $k=2$ photon-number components.

\section{Conclusion}

We have advanced in the problem of certifying non-classical correlations in prepare-and-measure scenarios. Addressing the problem from an adversarial perspective, we classify correlations as classical when the measurement outcomes can be predicted by a hypothetical adversary with access to the underlying deterministic mechanisms. This approach is inspired by the complementary QRNG problem: rather than estimating how predictable are the measurement outcomes, we ask whether the correlations admit a model with predictable outcomes. Thanks to this parallelism, we adapt already existing QRNG techniques for certifying non-classicality in the prepare-and-measure scenarios under various communication restrictions.

This approach reveals two clear conceptual advantages. Thus far, the problem of deciding whether correlations in the prepare-and-measure scenario admit a classical model has been only explored in dimension-restricted frameworks \cite{Poderini2020,degois2021,divi2023}. Our approach instead, can naturally inherit already existing QRNG methodology and adapt it to certify non-classicality under a wide range of communication constraints. Another conceptual consequence of our approach is that it reveals a nested structure of classicality in prepare-and-measure correlations. Conventionally, classicality is treated as a binary property: a correlation is either classical, in the sense that it admits a classical model for all settings, e.g.~\eqref{eq:SC}, or it is non-classical. Within our approach, however, classicality is associated with outcome predictability for specific choice of inputs. This reveals a nested classicality structure: a correlation can only be classical for some input choices, leading to a hierarchy of intermediate levels of non-classical behaviours.

\acknowledgements

The authors thank Leonardo Zambrano and Armin Tavakoli for useful feedback on the manuscript. This project has received funding from the European Union’s Horizon 2020 research and innovation programme under the Marie Skłodowska-Curie grant agreement No 101262877, Project TPSQCrypto. Views and opinions expressed are however those of the author(s) only and do not necessarily reflect those of the European Union or European Research Executive Agency. Neither the European Union nor the granting authority can be held responsible for them. N.D. and C.R.C. are financially supported by the Swedish Foundation for Strategic Research.

\bibliography{qrng_sd}

\appendix
\section{Semidefinite relaxations for classicality}
\label{app:relaxations}

The certification methods of the main text all reduce to a single decision problem. Given a set of observed correlations $p(b|x,y)$ and one of the communication assumptions $\mathcal{S}\in\{\mathcal{S}_\mathcal{E},\mathcal{S}_{\beta},\mathcal{S}_O\}$, we ask whether the data is compatible with a quantum model obeying $\mathcal{S}$, i.e. whether $p(b|x,y)\in\mathcal{Q}_\mathcal{S}$. Equivalently, we search for a hidden-variable decomposition $\{q(\lambda),\rho_x^\lambda,M_{b|y}^\lambda\}$ that reproduces $p(b|x,y)$ while fulfilling $\mathcal{S}$. This is a non-linear feasibility problem in the unknown states and measurements.
 
The strategy common to the three cases is to relax this problem to a hierarchy of semidefinite programs (SDPs) through a moment description. One collects a list $L$ of operators relevant to the scenario, closes it under products into a set $S\supseteq L$, and arranges the moments of the elements of $S$ into a matrix that is positive semidefinite by construction. The defining identities of the underlying quantum model, normalization and positivity of states and measurements, the assumption $\mathcal{S}$, and the observed statistics, all become linear relations among the entries of this matrix. The feasibility problem is thereby relaxed to an SDP, whose infeasibility certifies that no quantum model compatible with $\mathcal{S}$ can reproduce $p(b|x,y)$.
 
Two nested problems are addressed in this way. The first is membership in $\mathcal{Q}_\mathcal{S}$: feasibility of the SDP certifies that the correlations belong to an outer approximation $\mathcal{Q}_\mathcal{S}'\supseteq\mathcal{Q}_\mathcal{S}$, which can be tightened towards $\mathcal{Q}_\mathcal{S}$ by enlarging $S$. The second, and central to this work, is membership in the classical set $\mathcal{C}_\mathcal{S}^{(x,y)}$. By Observation~\ref{obs:pg_det}, the correlations admit a classical model for the selected inputs if and only if the measurement outcomes are perfectly predictable by the adversary, $p_g^{(x,y)}=1$. Imposing this predictability as an additional linear constraint on the moments yields an SDP whose infeasibility certifies the non-classicality of the observed correlations for those inputs. Throughout, $\lambda$ labels the hidden-variable. When predictability is imposed over a set of selected inputs $X^*\subseteq X,Y^*\subseteq Y$, it is convenient to let $\lambda=\{\lambda_{xy}\}_{x,y}$ run over the deterministic assignments that fix the outcome $b=\lambda_{xy}$ for each selected pair.
 
The three subsections below instantiate this program for each assumption, using the methodology best suited to it: a block-moment matrix of fixed operator dimension for the fixed ensemble (\secref{app:ensemble}), a Gram matrix of a purified model for bounded overlaps (\secref{app:overlaps}), and a tracial moment matrix for restricted observables (\secref{app:observables}).

\subsection{Fixed ensemble}
\label{app:ensemble}

Under the assumption $\mathcal{S}_\mathcal{E}$, the prepared ensemble is known and is supported by a Hilbert space of limited dimension $d$. The correlations of interest can then be generated by $d\times d$ operators that behave as valid quantum states and measurements. The decision problem then amounts to search for a valid decomposition $\rho_x=\sum_\lambda q(\lambda)\rho_x^\lambda$ that reproduces $p(b|x,y)$. To solve this, we invoke the block-moment matrix relaxation from Ref~\cite{dalessandro2026}, based on bounding the set of quantum correlations generated by operator-blocks of a fixed dimension $d$.
 
The method works as follows. Let $L^\lambda=\{ \mathds{1}, \{\rho_x^\lambda\}_x, \{M_{b|y}^\lambda\}_{b,y} \}$ be the list of operators relevant to the scenario under scrutiny for every possible value of $\lambda$. We then define the set $S^\lambda\supseteq L^\lambda$ containing products of elements in $L^\lambda$, through which we build the block-moment matrix
\begin{equation}
    \Gamma^\lambda = q_\lambda\sum_{u,v\in S^\lambda} uv^\dagger \otimes \ketbra{i_u}{i_v},
    \label{eq:bmm}
\end{equation}
where $i_u(i_v)$ defines the position of the element $u(v)$ in $S^\lambda$ and we have conveniently absorbed $q_\lambda\equiv q(\lambda)$ in the definition of $\Gamma^\lambda$. The matrix $\Gamma^\lambda$ is positive semidefinite by construction. Furthermore the observed correlations and relevant constraints on the problem can be expressed as linear combinations of elements of $\Gamma^\lambda$. To see this, consider the simplest case $S^\lambda=L^\lambda$ $\forall\lambda$. For any given $\lambda$, the block-moment matrix then becomes:
\begin{equation}
    \Gamma^\lambda=q_\lambda\begin{pmatrix}
        \mathds{1} & \rho_x^\lambda & M_{b|y}^\lambda\\[1mm]
        \cdot & \rho_x^\lambda\rho_{x'}^\lambda & \rho_x^\lambda M_{b|y}^\lambda\\[1mm]
        \cdot & \cdot & M_{b|y}^\lambda M_{b'|y'}^\lambda
    \end{pmatrix},
    \label{eq:bmm_level1}
\end{equation}
where only the upper diagonal was displayed since the matrix is Hermitian. Every element of $\Gamma^\lambda$ has to be understood as a block composed of $d\times d$ variables, which we index by the operators of $S^\lambda$ to which they correspond. To simplify notation, we denote $\Gamma^\lambda_{\rho_x,M_{b|y}}$ the block corresponding to $q_\lambda\,\rho_x^\lambda (M_{b|y}^\lambda)^\dagger$, and analogously for the remaining slots, with the hidden-variable dependence carried by the block superscript $\lambda$. These variables are used to mirror the properties of measurements and preparations. Especially for our purpose, we fix the average state using the linear equality $\sum_\lambda \Gamma^\lambda_{\rho_x, \mathds{1}}=\rho_x$, with $\Gamma^\lambda_{\rho_x, \mathds{1}}\succeq 0$ and $\text{tr}(\Gamma^\lambda_{\rho_x, \mathds{1}})=1$, following from state positivity and normalization. The relevant correlations appear as $p(\lambda,b|x,y)=\tr( \Gamma^\lambda_{\rho_x, M_{b|y}} )$. Positivity and normalization of the measurements are imposed through $\Gamma^\lambda_{M_{b|y}, \mathds{1}}\succeq 0$ and $\sum_{b,\lambda}\Gamma^\lambda_{M_{b|y}, \mathds{1}}=\mathds{1}$. Any set of correlations $p(b|x,y)\in\mathcal{Q}_{\mathcal{S}_{\mathcal E}}$ and average ensemble $\mathcal{E}=\{\rho_x\}_x$ is compatible with a set of realizable $\{\Gamma^\lambda\}_\lambda$. The membership problem can therefore be relaxed to the SDP feasibility problem,
\begin{align}\label{eq:sdp_ensemble}
\mathrm{find} 
& \quad \{\Gamma^\lambda\}_\lambda \\
\mathrm{s.t.}
& \quad \Gamma^\lambda\succeq 0, \nonumber \\
& \quad \sum_\lambda \Gamma^\lambda_{\rho_x, \mathds{1}}=\rho_x, \nonumber \\
& \quad p(b|x,y)=\sum_\lambda\text{tr}( \Gamma^\lambda_{\rho_x^, M_{b|y}} ), \nonumber
\end{align}
where the elements of $\Gamma^\lambda$ are constrained as stated above. Feasibility of the SDP in~\eqref{eq:sdp_ensemble} implies $p(b|x,y)\in\mathcal{Q}_{\mathcal{S}_\mathcal{E}}'$, with $\mathcal{Q}_{\mathcal{S}_\mathcal{E}}'\supseteq\mathcal{Q}_{\mathcal{S}_\mathcal{E}}$ denoting an outer approximation of the quantum set under the fixed-ensemble assumption. By including additional monomials in $S^\lambda$, one obtains a hierarchy of increasingly tight SDP relaxations that can be used to characterize $\mathcal{Q}_{\mathcal{S}_\mathcal{E}}$.
 
The same framework can be applied to the characterization of $\mathcal{C}_{\mathcal{S}_\mathcal{E}}^{(x,y)}$ for a selected choice of inputs $x\in X^\ast$, $y\in Y^\ast$. To that end, one only needs to accommodate the perfect-predictability condition of Observation~\ref{obs:pg_det}. For a particular input pair $(x,y)$ the hidden-variable can be identified with the predicted outcome, $\lambda=b$, and the condition amounts to adding $\sum_b \text{tr}( \Gamma^b_{\rho_{x^*}, M_{b|y^*}} )=1$. More generally, letting $\lambda=\{\lambda_{xy}\}_{x,y}$ index the deterministic assignments that set $b=\lambda_{xy}$ for every selected pair, predictability becomes the linear constraint $\sum_\lambda \text{tr}( \Gamma^\lambda_{\rho_x, M_{\lambda_{xy}|y}} )=1$. Collecting it together with the membership constraints, the decision problem for $\mathcal{C}_{\mathcal{S}_\mathcal{E}}^{(x,y)}$ takes the form of the feasibility program
\begin{align}
\text{find} & \quad \{\Gamma^{\bm{\lambda}}\}_{\bm{\lambda}}\\
\text{s.t.} & \quad \Gamma^{\bm{\lambda}}\succeq 0, \quad \sum_{\bm{\lambda}}\Gamma^{\bm{\lambda}}_{\rho_x,\mathds{1}}=\rho_x \nonumber \\
& \quad p(b|x,y)=\sum_{\bm{\lambda}} \tr(\Gamma^{\bm{\lambda}}_{\rho_x, M_{b|y}})\nonumber \\
& \quad \sum_{\bm{\lambda}} \tr(\Gamma^{\bm{\lambda}}_{\rho_x, M_{\lambda_{xy}|y}})=1, \quad x\in X^\ast, \ y\in Y^\ast \nonumber ,
\label{eq:appendix_4}
\end{align}
which is precisely the program quoted in the main text. Its infeasibility certifies the non-classicality of the observed correlations $p(b|x,y)$ for the selected inputs.

\subsection{Bounded overlaps}
\label{app:overlaps}

We now address the assumption $\mathcal{S}_\beta$, in which the preparations are pure, $\rho_x^\lambda=\ketbra{\psi_x^\lambda}{\psi_x^\lambda}$, and are characterized only through lower bounds on their averaged pairwise overlaps,
\begin{equation}
    \sum_\lambda q(\lambda)\,\lvert\braket{\psi_x^\lambda}{\psi_{x'}^\lambda}\rvert\ge\beta_{xx'}.
    \label{eq:overlap_assumption}
\end{equation}
Since no Hilbert-space dimension is assumed here, we resort to the Gram-matrix hierarchy of Ref.~\cite{wang2019}, which characterizes the correlations directly at the level of inner products between the vectors of a purified model.
 
We recall the purified description introduced in the main text. We consider the purified sates and measurement operators on the extended Hilbert space,
\begin{align}
    \ket{\Psi_x}&=\sum_\lambda\sqrt{q(\lambda)}, \ket{\psi_x^\lambda}\otimes\ket{\lambda}\otimes\ket{\lambda}_{\mathcal{H}_E},\label{eq:dilation_state}\\
    B_{b|y}&=\sum_\lambda M_{b|y}^\lambda\otimes\ketbra{\lambda}{\lambda},\qquad E_\lambda=\ketbra{\lambda}{\lambda},\label{eq:dilation_meas}
\end{align}
where $E_\lambda$ acts on the register $\mathcal{H}_E$ accessible to the adversary, so that $[B_{b|y},E_\lambda]=0$. With this dilation, the joint adversarial statistics read $p(\lambda,b|x,y)=\bra{\Psi_x}B_{b|y}\otimes E_\lambda\ket{\Psi_x}=q(\lambda)\,\text{tr}(\rho_x^\lambda M_{b|y}^\lambda)$.

The relaxation is built from the operator list $L=\{\,\mathds{1},\,\{B_{b|y}\otimes \mathds{1}_E\}_{b,y},\,\{\mathds{1}_B\otimes E_{\lambda}\}_\lambda\,\}$, collecting the vectors obtained by the action of products of elements of $L$ on the preparations $\{\ket{\Psi_x}\}_x$. For notational convenience, we write $B_{b|y}\otimes\mathds{1}_E\rightarrow B_{b|y}$ and $\mathds{1}_B\otimes E_{\lambda}\rightarrow E_{\lambda}$, and replace the explicit tensor-product structure with the commutation relations $[B_{b|y},E_{\lambda}]=0$ for all $b$, $y$ and $\lambda$. The Gram matrix $G$ is then defined through the entries $G_{u,v}=\braket{u}{v}$ for $u,v\in S$. By construction $G\succeq 0$. At first level, $S=\{\ket{\Psi_x}\}_x\cup\{B_{b|y}\ket{\Psi_x}\}_{b,x,y}\cup\{E_\lambda\ket{\Psi_x}\}_{\lambda,x}$ and the Gram matrix becomes
\begin{align}\label{eq:Gmat}
&G = \nonumber \\ 
&\begin{pmatrix}
\braket{\Psi_x}{\Psi_{x'}\!} & \bra{\Psi_x}\!B_{b'|y'}\!\ket{\Psi_{x'}\!}  & \bra{\Psi_x}\!E_{\lambda'}\!\ket{\Psi_{x'}\!} \\
\cdot &\!\! \bra{\Psi_x}\!B_{b|y}B_{b'|y'}\!\ket{\Psi_{x'}\!} \!\!&\!\!\! \bra{\Psi_x}\!B_{b|y}E_{\lambda'}\!\ket{\Psi_{x'}\!}\!\! \\
\cdot & \cdot & \bra{\Psi_x}\!E_{\lambda} E_{\lambda'}\!\ket{\Psi_{x'}\!}
\end{pmatrix}.
\end{align}
where each entry should be understood as an indexed block.

The operator identities of the model, namely $\sum_b B_{b|y}=\mathds{1}$, $B_{b|y}\succeq 0$, the orthogonality $E_\lambda E_{\lambda'}=\delta_{\lambda\lambda'}E_\lambda$ with $\sum_\lambda E_\lambda=\mathds{1}$, and the commutation $[B_{b|y},E_\lambda]=0$, all translate into linear relations among the entries of $G$. In particular, normalization of the dilated states gives $G_{\ket{\Psi_x},\ket{\Psi_x}}=1$, while the off-diagonal entries reproduce the averaged overlaps, $G_{\ket{\Psi_x},\ket{\Psi_{x'}}}=\sum_\lambda q(\lambda)\braket{\psi_x^\lambda}{\psi_{x'}^\lambda}$. The assumption~\eqref{eq:overlap_assumption} becomes the linear constraint $G_{\ket{\Psi_x},\ket{\Psi_{x'}}}\ge\beta_{xx'}$, where we considered that we can w.l.g. reduce the entries of the Gram matrix to be real and therefore ignore the absolute value. Measurement completeness imposes $\sum_b G_{u,B_{b|y}\ket{\Psi_x}}=G_{u,\ket{\Psi_x}}$ for all $u\in S$, the adversarial register satisfies $\sum_\lambda G_{u,E_\lambda\ket{\Psi_x}}=G_{u,\ket{\Psi_x}}$ and $G_{E_\lambda\ket{\Psi_x},E_{\lambda'}\ket{\Psi_{x'}}}=\delta_{\lambda\lambda'}\,G_{E_\lambda\ket{\Psi_x},\ket{\Psi_{x'}}}$, and the commutation $[B_{b|y},E_\lambda]=0$ translates into $G_{B_{b|y}\ket{\Psi_x},E_\lambda\ket{\Psi_{x'}}}=G_{E_\lambda\ket{\Psi_{x'}}, B_{b|y}\ket{\Psi_x}}$ $\forall \ x,x'$. The observed correlations finally appear as $p(b|x,y)=\bra{\Psi_x}B_{b|y}\ket{\Psi_x}=G_{\ket{\Psi_x},B_{b|y}\ket{\Psi_x}}$.
The membership problem is thus relaxed to the feasibility SDP
\begin{align}
\mathrm{find} 
& \quad G \label{eq:sdp_overlap} \\
\mathrm{s.t.}
& \quad G\succeq 0,\quad G_{\ket{\Psi_x},\ket{\Psi_{x'}}}\ge\beta_{xx'}, \nonumber \\
& \quad p(b|x,y)=G_{\ket{\Psi_x},B_{b|y}\ket{\Psi_x}}, \nonumber
\end{align}
whose feasibility certifies $p(b|x,y)\in\mathcal{Q}_{\mathcal{S}_{\beta}}'\supseteq\mathcal{Q}_{\mathcal{S}_{\beta}}$, an outer approximation of the quantum set compatible with the overlap assumption. Enlarging $S$ with higher-order products of $B_{b|y}$ and $E_\lambda$ yields a convergent hierarchy. 

To characterize $\mathcal{C}_{\mathcal{S}_{\beta}}^{(x,y)}$ we add the predictability condition of Observation~\ref{obs:pg_det}. The adversary guesses correctly when $\lambda=b$, with probability $G_{E_b\ket{\Psi_x},B_{b|y}\ket{\Psi_x}}=p(\lambda=b,b|x,y)$. Perfect predictability for the selected inputs, $p_g^{(x,y)}=1$, is therefore a unit guessing probability, $p_g^{(x,y)}=\sum_b G_{E_b\ket{\Psi_x},B_{b|y}\ket{\Psi_x}}=1$, and the decision problem for $\mathcal{C}_{\mathcal{S}_{\beta}}^{(x,y)}$ becomes
\begin{align}
\text{find} & \quad G \\
\text{s.t.} & \quad G\succeq 0, \quad G_{\ket{\Psi_x},\ket{\Psi_{x'}}} \geq \beta_{xx'} \nonumber\\
& \quad p(b|x,y) = G_{\ket{\Psi_x},B_{b|y}\ket{\Psi_x}} \nonumber\\
& \sum_b G_{E_b\ket{\Psi_x},B_{b|y}\ket{\Psi_x}}=1,\quad x\in X^*,\, y\in Y^*, \nonumber
\end{align}
which reproduces the program in the main text. Its infeasibility certifies the non-classicality of the observed correlations $p(b|x,y)$ for the selected inputs.

\subsection{Restricted observables}
\label{app:observables}
Finally, we treat the assumption $\mathcal{S}_O$, in which the preparations are constrained only through bounds on the expectation values of a set of Hermitian observables $\{O_{ix}\}_{i,x}$,
\begin{equation}
    \sum_\lambda q(\lambda)\,\text{tr}( \rho_x^\lambda O_{ix} )\le\omega_{ix}.
    \label{eq:observable_assumption}
\end{equation}
Since the assumptions, the correlations and the predictability conditions are all expectation values, it is natural to work directly with the trace functional through the tracial moment hierarchy of Refs.~\cite{pauwels2022,carceller2025_photon}.
 
The relaxation is built from the operator list $L=\{\,\mathds{1},\,\{\rho_x\}_x,\,\{M_{b|y}\}_{b,y},\,\{O_{ix}\}_{i,x}\,\}$, collecting the monomials obtained from products of its elements, $S\supseteq L$. For every value of $\lambda$ we define the tracial moment matrix through the scalar entries $\Upsilon^\lambda_{u,v}=q(\lambda)\,\text{tr}( uv^\dagger )$, for $u,v\in S$. Also in this case, the matrix is positive semidefinite by construction. At first level, $S=L$ and the moment matrix becomes
\begin{align}\label{eq:tracial_mat}
&\Upsilon^\lambda = \\
&q(\lambda)\begin{pmatrix}
\text{tr}\,(\mathds{1}) & \text{tr}\,(\rho_x^\lambda) & \text{tr}\,(M_{b|y}^\lambda) & \text{tr}\,(O_{ix}) \\
\cdot & \text{tr}(\rho_x^\lambda \rho_{x'}^\lambda) & \text{tr}(\rho_x^\lambda M_{b|y}^\lambda) & \text{tr}(\rho_x^\lambda O_{ix}) \\
\cdot & \cdot & \text{tr}(M_{b|y}^\lambda M_{b'|y'}^\lambda) & \text{tr}(M_{b|y}^\lambda O_{ix}) \\
\cdot & \cdot & \cdot & \text{tr}(O_{ix} O_{i'x'})
\end{pmatrix}, \nonumber
\end{align}
where, as before, each entry is indexed by the corresponding operators of $S$, $\Upsilon^\lambda_{u,v}$. The operator relations of the model are included as linear relations on the moments. Normalization gives $\Upsilon^\lambda_{\rho_x,\mathds{1}}=q(\lambda)$, measurement completeness gives $\sum_b\Upsilon^\lambda_{u,M_{b|y}}=\Upsilon^\lambda_{u,\mathds{1}}$ for all $u\in S$. The observable bounds~\eqref{eq:observable_assumption} read $\sum_\lambda\Upsilon^\lambda_{\rho_x,O_{ix}}\le\omega_{ix}$, while the observed statistics appear as $p(b|x,y)=\sum_\lambda\Upsilon^\lambda_{\rho_x,M_{b|y}}$. 
The membership problem is thus relaxed to the feasibility SDP
\begin{align}
\text{find} & \quad \{\Upsilon^\lambda\}_\lambda \\
\text{s.t.} & \quad \Upsilon^\lambda\succeq 0, \nonumber\\
& \quad \sum_\lambda\Upsilon^\lambda_{\rho_x,O_{ix}}\le\omega_{ix}, \nonumber\\
& \quad p(b|x,y)=\sum_\lambda\Upsilon^\lambda_{\rho_x,M_{b|y}}, \nonumber
\label{eq:sdp_observable}
\end{align}
whose feasibility certifies $p(b|x,y)\in\mathcal{Q}_{\mathcal{S}_O}'\supseteq\mathcal{Q}_{\mathcal{S}_O}$, an outer approximation of the quantum set compatible with the observable bounds. Enlarging $S$ with higher-order products yields an increasingly precise hierarchy of relaxations of the quantum set. 

To characterize $\mathcal{C}_{\mathcal{S}_O}^{(x,y)}$ we just need to add the predictability condition of Observation~\ref{obs:pg_det}. Letting $\lambda=\{\lambda_{xy}\}_{x,y}$ index the deterministic assignments that fix the outcome $b=\lambda_{xy}$ for every selected pair, perfect predictability becomes the linear constraint $\sum_\lambda\Upsilon^\lambda_{\rho_x,M_{\lambda_{xy}|y}}=1$, and the decision problem for $\mathcal{C}_{\mathcal{S}_O}^{(x,y)}$ becomes
\begin{align}
\text{find} & \quad \{\Upsilon^\lambda\}_\lambda \\
\text{s.t.} & \quad \Upsilon^\lambda\succeq 0, \nonumber\\
& \quad \sum_\lambda\Upsilon^\lambda_{\rho_x,O_{ix}}\le\omega_{ix} \nonumber\\
& \quad p(b|x,y) = \sum_\lambda\Upsilon^\lambda_{\rho_x,M_{b|y}} \nonumber\\
& \quad \sum_\lambda\Upsilon^\lambda_{\rho_x,M_{\lambda_{xy}|y}}=1,\quad x\in X^*,\, y\in Y^*, \nonumber
\end{align}
which reproduces the program in the main text. Its infeasibility certifies the non-classicality of the observed correlations $p(b|x,y)$ for the selected inputs.

\section{Quantum state discrimination with bounded overlaps}
\label{app.qsd}

In this Appendix, we derive the analytical bound for the success probability of quantum state discrimination of $n$ states compatible with predictable measurement outcomes for a single input $\abs{X^\ast}=1$.

Our goal is to solve the following optimization problem,
\begin{align}
    \max_{\{\rho_x^\lambda\},\{M_{b}^\lambda\}} & \quad \frac{1}{n}\sum_{x=1}^{n} \sum_\lambda q(\lambda) \tr\!\left(\rho_x^\lambda M_x^\lambda\right) \\
    \text{s.t.} & \quad \rho_x^\lambda = \ketbra{\psi_x^\lambda}{\psi_x^\lambda} , \ \sum_\lambda q(\lambda) \abs{\braket*{\psi_x^\lambda}{\psi_{x'}^\lambda}} \geq \beta , \nonumber \\
    & \quad \sum_\lambda q(\lambda) \tr\!\left(\rho_{x^\ast}^\lambda M_{\lambda}^\lambda\right) = 1 , \nonumber
\end{align}
where it is left implicit that the optimization runs over all valid quantum state operators $\rho_x^\lambda$ and POVMs $M_{b}^\lambda$, and for a concrete choice of $x^\ast\in X$. Given that the only relevant restriction on the state preparations is a lower-bound on the pair-wise overlap symmetrically imposed across all pairs of $x \neq x'$, we can without loss of generality fix e.g.~$x^\ast = 1$.

First of all, we note that at the optimum point, the overlap constraint saturates to $\abs{\!\braket*{\psi_x^\lambda}{\psi_{x'}^\lambda}}=\beta$, $\forall \lambda$. To see this, consider
\begin{align}
    &\max_{\{\rho_x^\lambda\},\{M_{b}^\lambda\}} \sum_\lambda q(\lambda) \frac{1}{n}\sum_{x=1}^{n} \tr\!\left(\rho_x^\lambda M_x^\lambda\right) \\
    &\leq \sum_\lambda q(\lambda) \max_{\{\rho_x^\lambda\},\{M_{b}^\lambda\}} \frac{1}{n}\sum_{x=1}^{n} \tr\!\left(\rho_x^\lambda M_x^\lambda\right) = \sum_\lambda q(\lambda) P(c_\lambda) \nonumber
\end{align}
where $P(c_\lambda)=\max_{\{\rho_x^\lambda\},\{M_{b}^\lambda\}} \frac{1}{n}\sum_{x=1}^{n} \tr\!\left(\rho_x^\lambda M_x^\lambda\right)$ (such that $\abs*{\!\braket*{\psi_x^\lambda}{\psi_{x'}^\lambda}}\geq c_\lambda$, $\forall x,x'$) corresponds to the Helstrom bound for the realization $\lambda$. Given that $P(c_\lambda)$ is concave on $c_\lambda$,
\begin{align}
    \sum_\lambda q(\lambda) P(c_\lambda) \leq P\left( \sum_\lambda q(\lambda) c_\lambda \right) = P(\beta) \ ,
\end{align}
and thus, the optimum point corresponds to the discrimination of $n$ states with the same overlap $\beta$, i.e.~the constraint $\abs{\!\braket*{\psi_x^\lambda}{\psi_{x'}^\lambda}}=\beta$ is saturated.

To this end, all imposed constraints are unitarily invariant. This means that, for each $\lambda$, we can rotate to a reference frame in which the preparations are independent of $\lambda$. Namely, we define the states $\ket{\phi_x}$ such that $\ket{\psi_x^\lambda}=U_\lambda\ket{\phi_x}$. Under this transformation, the optimization problem converts to
\begin{align}
    \max_{\{\phi_x\},\{M_{b}^\lambda\}} & \quad \frac{1}{n}\sum_{x=1}^{n} \sum_\lambda q(\lambda) \tr\!\left(\phi_x M_x^\lambda\right) \\
    \text{s.t.} & \quad \phi_x = \ketbra{\phi_x}{\phi_x} , \ \abs{\!\braket*{\phi_x}{\phi_{x'}}} \geq \beta , \nonumber \\
    & \quad \sum_\lambda q(\lambda) \tr\!\left(\phi_1 M_{\lambda}^\lambda\right) = 1 , \nonumber
\end{align}
where we defined the rotated measurement operators $M_b^\lambda = U_\lambda^\dagger M_b^\lambda U_\lambda$.

Next, let us expand the objective function as,
\begin{align}
    W_n = & \frac{q(1)}{n} \left[\tr\!\left(\phi_1 M_1^1\right) + \tr\!\left(\phi_2 M_2^1\right) + \cdots + \tr\!\left(\phi_n M_n^1\right)\right] \nonumber \\
    + & \frac{q(2)}{n} \left[\tr\!\left(\phi_1 M_1^2\right) + \tr\!\left(\phi_2 M_2^2\right) + \cdots + \tr\!\left(\phi_n M_n^2\right)\right] \nonumber \\
    \cdots + & \frac{q(n)}{n} \left[\tr\!\left(\phi_1 M_1^n\right) + \tr\!\left(\phi_2 M_2^n\right) + \cdots + \tr\!\left(\phi_n M_n^n\right)\right]. \nonumber
\end{align}
The perfect outcome predictability constraint implies that $M_{\lambda}^{\lambda}\ket*{\psi_1^\lambda}=\ket*{\phi_1}$.  Because the dimension of the Hilbert space is unconstrained, any support of $M_\lambda^\lambda$ orthogonal to $\ket*{\phi_1}$ can be reassigned to the remaining measurement outcomes without decreasing the objective function. Hence, there is always an optimal solution in which every measurement is rank-$1$ projective, $M_{\lambda}^\lambda = \ketbra*{\phi_1}{\phi_1}$. Given the completeness relation $\sum_b M_b^\lambda = \mathds{1}$, and the fact that the dimension is kept completely unconstrained, we may write $M_b^\lambda = \ketbra*{v_b^\lambda}{v_b^\lambda}$, where $\{\ket{v_b^\lambda}\}_b$ forms a complete basis with $\braket*{v_b^\lambda}{v_{b'}^\lambda} = \delta_{b,b'}$, $\forall \lambda$. With this notation then, $\ket{v_\lambda^\lambda}=\ket{\phi_1}$ is fixed. This reduces the objective function to
\begin{align}
    W_n = & \frac{q(1)}{n} \left[1 + \tr\!\left(\phi_2 M_2^1\right) + \cdots + \tr\!\left(\phi_n M_n^1\right)\right] \\
    + & \frac{q(2)}{n} \left[0 + \beta^2 + \cdots + \tr\!\left(\phi_n M_n^2\right)\right] \nonumber \\
    \cdots + & \frac{q(n)}{n} \left[0 + \tr\!\left(\phi_2 M_2^n\right) + \cdots + \beta^2\right]. \nonumber
\end{align}
Now, we note that the problem remains unchanged if the following label-permutation is applied,
\begin{align}
    &\text{for} \ \lambda = 1 \quad \left\{\begin{array}{ll}
        x\rightarrow x\oplus 1 & \text{if} \ x\neq1 \ \& \ x\oplus1\neq 1   \\
        x\rightarrow x\oplus 2 & \text{if} \ x\neq1 \ \& \ x\oplus1 = 1   \\ 
        x\rightarrow x & \text{if} \ x=1  
    \end{array}\right. \\
    &\text{for} \ \lambda > 1 \quad \left\{\begin{array}{ll}
        x\rightarrow x\oplus 1 & \text{if} \ x\neq1,\lambda \ \& \ x\oplus1\neq 1,\lambda   \\
        x\rightarrow x\oplus 2 & \text{if} \ x\neq1,\lambda \ \& \ x\oplus1 = 1,\lambda   \\ 
        x\rightarrow x & \text{if} \ x=1 ,\lambda 
    \end{array}\right. \nonumber ,
\end{align}
where $\oplus$ denotes a sum modulo $n$. Namely, the cyclic permutation of the indices $x$ that do not land on $x=1$, and additionally to $x=\lambda$ for $\lambda >1$. By convexity, the mixture of different objective functions under such label permutations is also a solution to the problem. This lets us work with an averaged objective function over all permutations, leading to
\begin{align}
    W_n = & \frac{q(1)}{n} \left[1 + (n-1)\alpha^2\right]  \\
    + & \frac{q(2)}{n} \left[\beta^2 + (n-2)\gamma^2\right] \nonumber \\
    \cdots + & \frac{q(n)}{n} \left[\beta^2 + (n-2)\gamma^2\right]. \nonumber ,
\end{align}
where $\alpha^2 = \frac{1}{(n-1)}\sum_{x>1}\tr\!\left(\phi_x M_x^1\right)$, and $\gamma^2 = \frac{1}{n-2}\sum_{x>1,x\neq\lambda}\tr\!\left(\phi_x M_x^\lambda\right)$.

Finally, we need to find what values $\alpha$ and $\gamma$ can take. Let us focus on $\alpha$. For $\lambda=1$, the set $\{\ket{v_b^1}\}_b$ forms a complete orthonormal basis. In this basis, we can write
\begin{align}
    \ket{\psi_2^1} &= \beta \ket{v_1^1} + \alpha \ket{v_2^1} + g\sum_{j\neq 1,2}\ket{v_j^1} \\
    \ket{\psi_3^1} &= \beta \ket{v_1^1} + \alpha \ket{v_3^1} + g\sum_{j\neq 1,3}\ket{v_j^1}  \ , \nonumber
\end{align}
for $g = \sqrt{\frac{1-\beta^2-\alpha^2}{n-2}}$. The constraint $\braket{\phi_2}{\phi_3}=\beta$ can only be satisfied if and only if
\begin{align}\label{eq:alpha}
    \alpha = \frac{\sqrt{1-\beta}\left((n-2)+\sqrt{1+(n-1)\beta}\right)}{n-1} .
\end{align}
Using the same idea for $\lambda >1$, one concludes that $\alpha = \gamma$. At the end of the day thus, one finds that
\begin{align}
    W_n &\leq \frac{1}{n}\max\left\{1+(n-1)\alpha^2,\beta^2+(n-2)\gamma^2\right\} \\
    &= \frac{1}{n}\left(1+(n-1)\alpha^2\right) \ . \nonumber
\end{align}
Substituting $\alpha$ from \eqref{eq:alpha} yields the desired bound.

\end{document}